\documentclass{article}

\usepackage[dvips]{graphicx}
\usepackage{amsmath}
\usepackage{IEEEtrantools}
\usepackage{color}
\usepackage{amsthm,amsfonts}

\newcommand{\ket}[1]{\left| #1 \right>} 
\newcommand{\bra}[1]{\left< #1 \right|} 
\newcommand{\braket}[2]{\left< #1 \vphantom{#2} \right|
 \left. #2 \vphantom{#1} \right>} 
 
\newcommand{\redout}[1]{\textcolor{black}{#1}}
 
\newtheorem{definition}{Definition}
\newtheorem{theorem}{Theorem}
\newtheorem{lemma}[theorem]{Lemma}
\newtheorem{corollary}[theorem]{Corollary}

\title{Quantum Walks on the Line\\ with Phase Parameters}
\author{Marcos Villagra\thanks{marcos.villagra@acm.org, Graduate School of Information Science, Nara Institute of Science and Technology, Nara 630-0192, Japan.} \and Masaki Nakanishi\thanks{m-naka@e.yamagata-u.ac.jp, Faculty of Education, Art and Science, Yamagata University, Yamagata 990-8560, Japan.} \and Shigeru Yamashita\thanks{ger@cs.ritsumei.ac.jp, Department of Computer Science, Ritsumeikan
University, Shiga 525-8577, Japan.} \and Yasuhiko Nakashima\thanks{nakashim@is.naist.jp, Graduate School of Information Science, Nara Institute of Science and Technology, Nara 630-0192, Japan.}}
\date{}

\begin{document}
\maketitle
\begin{abstract}
In this paper, a study on discrete-time coined quantum walks on the line is presented. Clear mathematical foundations are still lacking for this quantum walk model. As a step \redout{towards} this objective, the following question is being addressed: {\it Given a graph, what is the probability that a quantum walk arrives at a given vertex after some number of steps?} This is a very natural question, and for random walks it can be answered by several different combinatorial arguments. For quantum walks this is a highly non-trivial task. Furthermore, this was only achieved before for one specific coin operator (Hadamard operator) for walks on the line. Even considering only walks on lines, generalizing these computations to a general $SU(2)$ coin operator is a complex task. The main contribution is a closed-form formula for the amplitudes of the state of the walk (which includes the question above) for a general symmetric $SU(2)$ operator for walks on the line. To this end, a coin operator with parameters that alters the phase of the state of the walk is defined. Then, closed-form solutions are computed by means of Fourier analysis and asymptotic approximation methods. We also present some basic properties of the walk which can be deducted using weak convergence theorems for quantum walks. In particular, the support of the induced probability distribution of the walk is calculated. Then, it is shown how changing the parameters in the coin operator affects the resulting probability distribution.
\end{abstract}


\section{Introduction}
The design of quantum algorithms is nowadays one of the major problems in the quantum computing community. The strategies for writing classical algorithms as divide and conquer, dynamic programming, etc, are not easily adapted to the quantum paradigm. Strategies for designing quantum algorithms are phase amplification, phase estimation, to name a few. As an example of the applications of these strategies, Grover's algorithm uses the amplitude amplification technique, and Shor's algorithm relies in reductions to order finding and phase estimation \cite{nielsen00}. Therefore, it becomes necessary the study of different approaches to improve the efficiency of the search.

One of the emergent alternatives for the design of algorithms are quantum walks, in direct analogy to random walks in classical computing. Random walks showed to be a successful tool for designing algorithms, and the same success is expected in the quantum paradigm. Results in this field showed that quantum walks can outperform its classical counterpart by exploiting quantum mechanical effects such as interference and superposition, giving and exponential speedup for certain types of graphs, and polynomial speedup for some practical applications \cite{ambainis04,kempe03}.

There are two types of discrete-time quantum walks, {\it Quantum Markov Chains} and {\it Coined Quantum Walks}. This paper is about the latter, and from now on we will refer to this model simply as quantum walk when is obvious from the context.

The field of quantum walks is very recent, and still lacks a solid mathematical foundation. Markov chain quantum walks already started to build these foundations by establishing a direct connection to classical Markov chains using algebraic techniques \cite{santha08}. However, coined quantum walks are not having the same luck, and it seems that mathematical techniques for random walks simply do not work.

\subsection{Related Work}
Coined Quantum Walks are defined by the application of two unitary operators ${S}$ and ${C}$, where ${C}$ (coin operator) decides which vertex to move onto, and ${S}$ (shift operator) performs the actual movement of the walk given the direction decided by $C$. Ambainis \cite{ambainis04}, Kempe \cite{kempe03} and Konno \cite{Konno2008} give good surveys of this model. There are several studies of this walk for specific graphs. On the line, Ambainis et al. \cite{ambainis01} and Chandrashekar,  Srikanth, and Laflamme \cite{chandrashekar08} show that the variance of the induced probability distribution has a quadratic improvement over the classical walk (i.e. for $t$ steps, $V=O(t^2)$ and classically $V=O(t)$). Konno computed the induced probability distribution using path integrals \cite{Konno2002} and via a weak limit theorem \cite{Konno2005}. In the hypercube, Kempe \cite{kempe05} shows that the hitting time from one corner to its opposite is exponentially faster, while Moore and Russell \cite{moore02} gives the same speed-up for the mixing time. For practical applications there are algorithms for hypercubes and grids. For the hypercube, Shenvi et al. \cite{shenvi03} gives an algorithm for solving SAT with a quadratic improvement, while Poto\v cek et al. \cite{potocek09} gives an improvement of the same algorithm on the success probability. For grids, Ambainis et al. \cite{ambainis05} show a quadratic speed-up and presents a general framework for analyzing quantum walks. Also, Ambainis \cite{ambainis07} gives an optimal algorithm for element distinctness over the Jhonson graph with a quadratic speed-up.

Quantum walks on the line is probably the most studied quantum walk model. Interest on this matter started in computer science with Ambainis, Bach, Nayak, Vishwanath, and Watrous \cite{ambainis01}, where notions of hitting and mixing times were introduced. In the same piece of work, they computed a closed-form formula for the induced probability distribution of a Hadamard walk (i.e. a quantum walk with a Hadamard operator as coin). Furthermore, their formula gives a complete characterization of the amplitudes in the state of the walk in the asymptotic limit.


It is known that the dynamics of the walk is controlled by the coin operator \cite{kempe03}. Thus, depending on the application, a good choice of the coin could make a great difference. This motivated the study of quantum walks on the line moved by a general $SU(2)$ operator, which has four independent variables. However if we consider only the resulting probability distribution, one variable is enough; i.e. any probability distribution resulting from a quantum walk on the line can be simulated by a general rotation around the $z$ axis with parameter $\theta$. Nayak and Vishwanath \cite{Nayak2000} gave an intuitive description of the probability distribution based on the stationary phase method without giving an explicit formula for it, and without considering the amplitudes of the state of the walk. Chandrashekar, Srikanth, and Laflamme \cite{chandrashekar08} studied generalized walks using a $SU(2)$ coin operation. They present an approximate formula for the amplitudes of the state of the walk. However, their results were based in numerical experiments rather than a complete analytically deducted formula. Grimmet, Janson, and Scudo \cite{grimmett04} showed a ballistic spreading of the walk and they gave an expression for the limit distribution using weak convergence theorems.

\subsection{Contributions}

As a step toward finding mathematical foundations of quantum walks, in this paper the following question is being addressed: {\it Given a graph, what is the probability that a quantum walk arrives at a given vertex after some number steps?} This is a very natural question, and for random walks it can be answered by several different combinatorial arguments \cite{stroock05}.


The main contribution of this paper is a closed-form formula\footnote{A quantity $f(n)$ is in closed-form if we can compute it using at most a fixed number of ``well-known" standard operations, independent of $n$ \cite{Graham1994}.} for the question above for a general symmetric $SU(2)$ operator for walks on the line (Theorem \ref{the:amplitudes}). Furthermore, the formula characterizes the amplitudes of the state of the walk in the asymptotic limit. In comparison to the previous works mentioned before (Nayak and Vishwanath \cite{Nayak2000}, Chandrashekar et al. \cite{chandrashekar08}), the closed-form formulas derived in this paper were analytically computed for the amplitudes of the state of the walk (including the induced probability distribution) for a symmetric $SU(2)$ operator (Table \ref{tab:known-results} shows more clearly these differences). Also, in a seminal work, Konno \cite{Konno2002,Konno2005} gave explicit expressions for the amplitudes of a $U(2)$ coin, using a discrete path integral method in a clever way. However, these expressions were not in closed-form, as we claim in this work. Furthermore, we show how to compute the errors in the asymptotic approximation, something that was missing from previous works in the literature. To this end, in Section \ref{sec:phase-parameters} a coin operator with parameters that alters the phase of the state of the walk on the line is proposed. The coin operator is inspired by the quantum algorithm for SAT proposed by Hogg \cite{Hogg2000}.  In that work, in order to implement heuristics for quantum algorithms, the author proposed to add parameters to the unitary operation of a search algorithm. This way, the situation is similar to classical algorithms where a tunable set of parameters are adjusted according to the problem. After defining the coin operation, we compute the  spectrum of the unitary evolution operator of the walk using Fourier analysis. In Section \ref{sec:approximation}, after having obtained the eigenspectrum of the walk,  we apply the inverse Fourier transform to obtain the state of the walk in terms of Fourier coefficients. To compute a closed-form solution in the asymptotic limit from the Fourier coefficients, we applied the Euler-Maclaurin formula \cite{Apostol1999} and the steepest descent method for asymptotic approximation of integrals \cite{wong01}. This method is in fact stronger than the stationary point method from \cite{Nayak2000} and \cite{Stefanak2008}, where the authors use it to study the asymptotics of the resulting probability distribution from coin operators with real eigenvalues. With the steepest descent method we can compute the amplitudes of the state of the walk resulting from any complex unitary operator.  In Section \ref{sec:formulas-convergence}, we compute the error terms for the approximations made, which can be derived from the employed methods. Finally, some basic properties of the walk are examined by means of weak convergence theorems \cite{grimmett04}. The support of the induced probability distribution of the walk is computed, and then we argue how changing the parameters in the coin operator affects the resulting  probability distribution. 

\begin{table}[t]
	\centering
	\caption{Known results for different coins for walks on the line.}
	\label{tab:known-results}
	\begin{tabular}{|p{2cm}|p{4cm}|p{4cm}|}
		\hline
		{\bf Coin} 		& {\bf Amplitudes of the state}				& {\bf Probability distribution}\\
		\hline
		\hline
		Hadamard			& closed-form \cite{Nayak2000}		& closed-form \cite{Nayak2000,Konno2002}\\
		\hline
		$SU(2)$			& numerical results \cite{chandrashekar08}	& numerical results \cite{chandrashekar08}, closed-form \cite{Nayak2000}\\
		\hline
		Symmetric $SU(2)$	& closed-form [this work]				& closed-form [this work]\\
		\hline
		$U(2)$			& explicit formula (not closed-form) \cite{Konno2002,Konno2005}		& explicit formula (not closed-form) \cite{Konno2005}\\
		\hline
	\end{tabular}
\end{table}

\section{Quantum Walks with Phase Parameters}\label{sec:phase-parameters}
In this section, a coin operation with parameters is proposed. Then, using Fourier analysis, integral-forms for the amplitudes of the walk on the line are computed. Later in Section \ref{sec:approximation}, it is shown how to solve these integrals and derive a closed-form in the asymptotic limit.

\subsection{Walks on the Line with Phase Parameters}
Here we define quantum walks on the line, and introduce the coin operator used in this research.

\begin{definition}
Let $\mathcal{H}_c=span\{\ket{\gets},\ket{\to}\}$ and $\mathcal{H}_s=span\{\ket{n}:n\in\mathbb{Z}\}$. The state of the walk  $|\Psi_t\rangle=\sum_n|\psi_t(n)\rangle$ at time $t$ is defined over the joint space $\mathcal{H}_c \otimes \mathcal{H}_s$ with basis states $\{|d,n\rangle : |d\rangle \in \mathcal{H}_c, |n\rangle \in\mathcal{H}_s\}$, where $|\psi_t(n)\rangle=\sum_{d}\alpha_t^d(n)|d,n\rangle$ and $\alpha_t^d(n)$ is the amplitude at time $t$ in direction $d$ and position $n$. Also $\sum_{d,n}|\alpha_t^d(n)|^2=1$.
\end{definition}

For the analysis of the walk on the line we consider the projection at time $t$ onto position $n$ as a 2 dimensional vector, i.e.
\[
\begin{bmatrix}
\alpha_t^\gets(n)\\
\alpha_t^\to(n)
\end{bmatrix}
\]
with $\alpha_t^\gets(n)$ and $\alpha_t^\to(n)$ representing the amplitude of the walker at position $n$ at time $t$ going left and right respectively. The probability of being at position $n$ at time $t$ is thus given by
\begin{equation}\label{eq:prob}
P_t(n)=|\langle \psi_t(n)|\psi_t(n)\rangle|^2=|\alpha_t^\gets(n)|^2+|\alpha_t^\to(n)|^2.
\end{equation}
Throughout the paper, the initial condition is considered as $|\psi_0(0)\rangle=[\alpha_0^\gets,\alpha_0^\to]^T$ and $|\psi_0(n)\rangle=[0,0]^T$ for $n\neq 0$, with $|\alpha_0^\gets|^2+|\alpha_0^\to|^2=1$.

The quantum walk is defined by the way it moves at each time step. This is captured by the following definition.

\begin{definition}
The time evolution of the walk is given by
\[ \ket{\Psi_t}=U\ket{\Psi_{t-1}},\text{ or equivalently, } \ket{\Psi_t}=U^t\ket{\Psi_0},\]
where $U=S(C\otimes I)$ is a unitary operator defined on the Hilbert space of the whole system $\mathcal{H}_c\otimes \mathcal{H}_s$, $I$ is the identity matrix acting on $\mathcal{H}_s$, $C$ is the coin operator acting solely on $\mathcal{H}_c$, and $S$ is the shift operator in charge of performing the walk.
\end{definition}

According to this definition, the walk first choses a direction of movement using $C$, and then moves with operator $S$. In order to move, operator $S$ needs to be conditioned on the coin space in the following way,
\begin{equation}
S=\sum_{n} \ket{\gets}\bra{\gets} \otimes \ket{n-1}\bra{n}+\ket{\to}\bra{\to} \otimes \ket{n+1}\bra{n}.
\end{equation}

\begin{definition}\label{def:coin}
The coin operator is defined by ${C}={H}{T}{H}$, where ${H}$ is the Hadamard operator\footnote{The Hadamard operator is defined as $H=\frac{1}{\sqrt{2}}\begin{bmatrix} 1 & 1\\ 1 & -1\end{bmatrix}$.} in charge of mixing amplitudes among states, and ${T}=e^{i\pi\tau_1}|\gets\rangle\langle \gets|+e^{i\pi\tau_2}|\to\rangle\langle \to|$ is the diagonal phase adjustments with $\tau_1,\tau_2\in[0,1]$.
\end{definition}

Let $a\equiv e^{i\pi\tau_1}+e^{i\pi\tau_2}$ and $b\equiv e^{i\pi\tau_1}-e^{i\pi\tau_2}$. Then, the resulting operator can be written as
\[
C=\frac{1}{2}\begin{bmatrix}
a	&	b\\
b	&	a
\end{bmatrix},
\]
which have the following effect on $\mathcal{H}_C$
\[\begin{array}{l}
| \gets\rangle \longrightarrow (1/2)a| \gets\rangle+(1/2)b| \to\rangle,\\
|\to\rangle \longrightarrow (1/2)b| \gets\rangle+(1/2)a| \to\rangle.
\end{array}\]

Figure \ref{fig:walk} shows the dynamics of a walk using ${C}$ as coin.
\begin{figure*}[t]
\center
\includegraphics[scale=0.7]{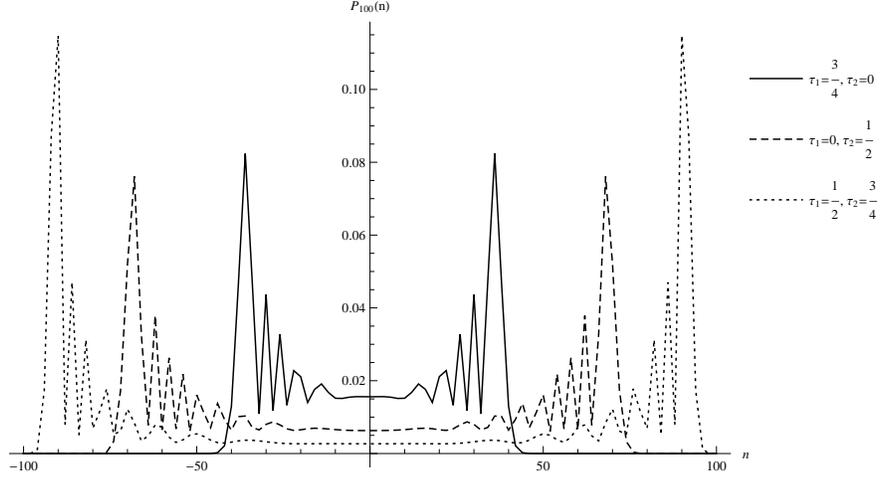}
\caption{Quantum walk on the line with different values of phase parameters. The variance of the walk changes depending on $\tau_1$ and $\tau_2$. Since the probabilities at odd positions are 0, those points are not plotted.}
\label{fig:walk}
\end{figure*}
For different values of the phase parameters $\tau_1$ and $\tau_2$ the variance of the induced probability distribution changes. 

The state of the walk at time $t$ can be related to the state at time $t+1$ according to the following lemma.
\begin{lemma}
\begin{equation}\label{eq:evolution}
|\psi_{t+1}(n)\rangle=M_+|\psi_t(n-1)\rangle+M_-|\psi_t(n+1)\rangle
\end{equation}
where
\[
M_+=\begin{bmatrix}
0		&	0\\
(1/2)b	&	(1/2)a
\end{bmatrix}
\text{ and }
M_-=\begin{bmatrix}
(1/2)a	&	(1/2)b\\
0		&	0
\end{bmatrix}.
\]
\end{lemma}
\begin{proof}
Let $|\Psi_t\rangle=\sum_n \alpha_t^\gets(n)|\gets,n\rangle+\alpha_t^\to(n)|\to,n\rangle$ be the state at time $t$. Also denote the amplitudes after applying operators $C$ and $S$ as
\begin{align*}
	(C\otimes I)\ket{\Psi_t}	&=\sum_n \alpha_t^\gets(n)'|\gets,n\rangle+\alpha_t^\to(n)'|\to,n\rangle,\\
	S(C\otimes I)|\Psi_t\rangle&=\sum_n \alpha_t^\gets(n)''|\gets,n\rangle+\alpha_t^\to(n)''|\to,n\rangle.
\end{align*}
Now let $|\Psi_{t+1}\rangle=\begin{bmatrix} \alpha_{t+1}^\gets(n)\\ \alpha_{t+1}^\to(n) \end{bmatrix}$
be the state at time $t+1$. The amplitudes of this state are related to the amplitudes of $|\Psi_t\rangle$ in the following way
\[
\begin{bmatrix}
	\alpha_{t+1}^\gets(n)\\
	\alpha_{t+1}^\to(n)
\end{bmatrix}
=
\begin{bmatrix}
	\alpha_{t}^\gets(n)''\\
	\alpha_{t}^\to(n)''
\end{bmatrix}
=
\begin{bmatrix}
	\alpha_{t}^\gets(n+1)'\\
	\alpha_{t}^\to(n-1)'
\end{bmatrix}.
\]

The contributions to the amplitudes of state $|\Psi_{t+1}\rangle$ come from position $n+1$ for the upper component, and from $n-1$ for the lower component by definition of operator $S$. The amplitudes corresponding to the state after applying $C$ are computed as follows:
\begin{IEEEeqnarray}{rCl}
	\IEEEeqnarraymulticol{3}{l}{C|\psi_t(n+1)\rangle}\nonumber\\ \qquad
					&=&\begin{bmatrix}
						(1/2)a\alpha_t^\gets(n+1)+(1/2)b\alpha_t^\to(n+1)\\
						(1/2)b\alpha_t^\gets(n+1)+(1/2)a\alpha_t^\to(n+1)
					\end{bmatrix}\nonumber\\
					&=&\begin{bmatrix}
						\alpha_{t}^\gets(n+1)'\\
						\alpha_{t}^\to(n+1)'
					\end{bmatrix},\nonumber
\end{IEEEeqnarray}
and the same for $C\ket{\psi_t(n-1)}$. Thus
\begin{align*}
|\psi_{t+1}(n)\rangle	
			&=\begin{bmatrix}
				(1/2)a\alpha_t^\gets(n+1)+(1/2)b\alpha_t^\to(n+1)\\
				(1/2)b\alpha_t^\gets(n-1)+(1/2)a\alpha_t^\to(n-1)
			\end{bmatrix}\\
			&=M_+|\psi_t(n-1)\rangle+M_-|\psi_t(n+1)\rangle,
\end{align*}
where
\[
M_+=\begin{bmatrix}
0		&	0\\
(1/2)b	&	(1/2)a
\end{bmatrix}
\text{ and }
M_-=\begin{bmatrix}
(1/2)a	&	(1/2)b\\
0		&	0
\end{bmatrix}.
\]
\end{proof}

\subsection{Analysis}\label{sec:analysis}
One approach to the analysis of quantum processes is the path integral approach. This method explicitly computes the amplitude of a certain state as the sum over all possible paths leading to that state \cite{ambainis01,Konno2002}. Solving a path integral is known to be hard, and we avoid this by following the steps of \cite{ambainis01,moore02,kempe05} known as the {\it Schr\"odinger approach}. Given the translational invariance of the walk, it has a simple description in Fourier space \cite{ambainis01}. The Fourier transform of the walk is analyzed and then transformed back to the original domain.

The quantum Fourier transform \cite{nielsen00} of a wave equation is defined by
\begin{equation}\label{eq:fourier}
\ket{\widetilde{\psi}_t(k)}=\sum_n e^{ikn} \ket{\psi_t(n)},
\end{equation}
and the corresponding inverse Fourier transform is then
\begin{equation}\label{eq:inversefourier}
\ket{\psi_t(n)}=\frac{1}{2\pi}\int_{-\pi}^{\pi} e^{-ikn} \ket{\widetilde{\psi}_t(k)} dk.
\end{equation}

Applying (\ref{eq:fourier}) to (\ref{eq:evolution}) we get
\begin{align*}
\ket{\widetilde{\psi}_{t+1}(k)}	&=\sum_n e^{ikn}M_+|\psi_t(n-1)\rangle+e^{ikn}M_-|\psi_t(n+1)\rangle\\
						&=e^{ik}M_+\sum_n e^{ik(n-1)}\ket{\psi_t(n-1)}\\
						&\qquad+e^{-ik}M_-\sum_n e^{ik(n+1)}\ket{\psi_t(n+1)}\\
						&=e^{ik}M_+\ket{\widetilde{\psi}_t(k)}+e^{-ik}M_-\ket{\widetilde{\psi}_t(k)}\\
						&=\left(e^{ik}M_+ + e^{-ik}M_-\right)\ket{\widetilde{\psi}_t(k)}.
\end{align*}
Then, the time-evolution in Fourier space is given by
\begin{equation}
\ket{\widetilde{\psi}_{t+1}(k)}=M_k\ket{\widetilde{\psi}_t(k)}
\end{equation}
where ${M}_k=e^{ik}M_+ +e^{-ik}M_-$. In matrix form
\begin{equation}\label{eq:operator}
{M}_k=\frac{1}{2}\begin{bmatrix}
ae^{-ik}	&	be^{-ik}\\
be^{ik}	&	ae^{ik}
\end{bmatrix}.
\end{equation}

In general, the state at time $t$ is given by the $t$-th power of operator ${M}_k$ applied to the initial state
\begin{equation}\label{eq:fourier-time-evolution}
\ket{\widetilde{\psi}_t(k)}= {M}_k^t\ket{\widetilde{\psi}_0(k)}.
\end{equation}

The following lemma shows the eigenspectrum of operator $M_k$.
\begin{lemma}\label{the:spectrum}
Let $M_k$ be a unitary matrix as in (\ref{eq:operator}). The eigenvalues and eigenvectors of $M_k$ are
\[
\lambda_j(k)=1/2\left( a\cos k \pm \sqrt{b^2-a^2\sin^2k}\right)
\]
and
\[
|\lambda_j(k)\rangle=N_j(k)\begin{bmatrix}
-ia\sin k \pm \sqrt{b^2-a^2\sin^2k}\\
be^{ik}
\end{bmatrix}
\]
respectively, with $j=1,2$. Furthermore, $N_j(k)$ is a normalization coefficient given by
\[
N_j(k)=\left( \left|-ai\sin k \pm \sqrt{b^2-a^2\sin^2 k}\right|^2+\left| b\right|^2 \right)^{-1/2}.
\]
\end{lemma}
\begin{proof}
The characteristic polynomial of $M_k$ is determined by $det(M_k-\lambda I)=0$. Then
\begin{align*}
det(M_k-\lambda I)	
				&=\lambda^2-a \lambda \cos k+\frac{a^2}{4}-\frac{b^2}{4}.
\end{align*}
Solving the equation gives the eigenvalues
\begin{align*}
\redout{\lambda_j(k)}	
			&=\frac{a \cos k \pm \sqrt{b^2-a^2\sin^2 k}}{2},
\end{align*}
for $j=1,2$.
In order to find the eigenvectors, we solve the following system of linear equations
\[
\left(M_k-\redout{\lambda_j(k)} I\right)\begin{bmatrix}
x_j\\
y_j
\end{bmatrix}
=
\begin{bmatrix}
x_j\left(\frac{a}{2}e^{-ik}-\redout{\lambda_j(k)}\right)+y_j \frac{b}{2}e^{-ik}\\
x_j\frac{b}{2}e^{ik}+y_j\left(\frac{a}{2}e^{ik}-\redout{\lambda_j(k)}\right)
\end{bmatrix}
=\begin{bmatrix}
0\\
0
\end{bmatrix}.
\]
By letting $y_j=1$, we get $x_j=(-a+2\lambda_je^{-ik})/{b}$. Given that any multiple of this vector is still an eigenvector, multiply $y_j$ and $x_j$ by $be^{ik}$ and obtain
\[
be^{ik}\begin{bmatrix}
	x_j\\
	y_j
\end{bmatrix}
=
\begin{bmatrix}
-ae^{ik}+2\lambda_j\\
be^{ik}
\end{bmatrix}
=
\begin{bmatrix}
-ai\sin k \pm\sqrt{b^2-a^2\sin^2 k} \\
be^{ik}
\end{bmatrix}.
\]
Then $N_j(k)$ is 1 divided by the $\ell_2$-norm of this vector, and multiply the eigenvectors by $N_j(k)$ to normalize them.
\end{proof}

Diagonalize (\ref{eq:operator}) to obtain
\[
M_k^t = \sum_{j\in\{1,2\}}\lambda_j(k)^t|\lambda_j(k)\rangle\langle \lambda_j(k)|,
\]
where $\lambda_1(k)$ and $\lambda_2(k)$ are the eigenvalues with corresponding eigenvectors $\ket{\lambda_1(k)}$ and $\ket{\lambda_2(k)}$. Now apply the diagonalized operator to the time evolution (\ref{eq:fourier-time-evolution}) and obtain the following form
\begin{align}
|\widetilde{\psi}_t(k)\rangle
				&=\sum_j\left(\lambda_j(k)^t\ket{\lambda_j(k)}\bra{\lambda_j(k)}\right)\ket{\widetilde{\psi}_0(k)}\nonumber\\
				&=\sum_j\braket{\lambda_j(k)}{\widetilde{\psi}_0(k)}\lambda_j(k)^t\ket{\lambda_j(k)}.\label{eq:time-evolution}
\end{align}

The initial state is $[\alpha_0^\gets,\alpha_0^\to]^T$, and in Fourier space becomes $|\widetilde{\psi}_0(k)\rangle=[{\alpha}_0^\gets,{\alpha}_0^\to]^T$ for all $k \in[-\pi,\pi]$. To write equation (\ref{eq:time-evolution}) in a simpler way, define
\begin{align}
\xi_j(k)	&=\braket{\lambda_j(k)}{\widetilde{\psi}_0(k)}\nonumber\\
		&=\alpha_0^\gets N_j(k)\left(-ia\sin k \pm \sqrt{b^2-a^2\sin^2k}\right)^*+\alpha_0^\to N_j(k)b^*e^{-ik},
\end{align}
where $*$ is the complex conjugate. This can be expressed in matrix form as
\begin{align*}
\begin{bmatrix}
\xi_1(k)\\
\xi_2(k)
\end{bmatrix}
&=
\begin{bmatrix}
(-ia \sin k+\sqrt{b^2-a^2\sin^2k})^*	& b^* e^{-ik}N_1(k)\\
(-ia \sin k-\sqrt{b^2-a^2\sin^2k})^*	& b^* e^{-ik}N_2(k)
\end{bmatrix}
\cdot
\begin{bmatrix}
\alpha_0^\gets\\
\alpha_0^\to
\end{bmatrix}.
\end{align*}

The state of the walk at time $t$ can be expressed by
\begin{equation}\label{eq:fourier-evolution}
\ket{\widetilde{\psi}_t(k)}=M_k^t\ket{\widetilde{\psi}_0(k)}=\sum_j \lambda_j^t(k) \xi_j(k)\ket{\lambda_j(k)}.
\end{equation}
Let $\widetilde{\alpha}_t^\gets(k)$ and $\widetilde{\alpha}_t^\to(k)$ be the amplitudes of the state $\ket{\widetilde{\psi}_t(k)}$ in Fourier space going left and right respectively. Then, by equation (\ref{eq:fourier-evolution}) and Lemma \ref{the:spectrum} these amplitudes  are
\begin{align}
\widetilde{\alpha}_t^\gets(k)	&=\sum_j \lambda_j(k)^t \xi_j(k) N_j(k)
						\left(-ia\sin k \pm \sqrt{b^2-a^2\sin^2 k}\right)\label{eq:fourierleft}
\end{align}
and
\begin{equation}
\widetilde{\alpha}_t^\to(k)=\sum_j\lambda_j(k)^t \xi_j(k) N_j(k) be^{ik}.\label{eq:fourierright}
\end{equation}

The final step is to reverse back to the original domain of the walk. This is done by applying (\ref{eq:inversefourier}) to (\ref{eq:fourierleft}) and (\ref{eq:fourierright}),
\begin{align}
\alpha_t^\gets(n)	&= \frac{1}{2\pi}\int_{-\pi}^{\pi}  \sum_j  \xi_j(k) N_j(k)\lambda_j^te^{-ikn}
				 \left(-ia\sin k\pm \sqrt{b^2-a^2\sin^2 k}\right)dk  \label{eq:leftintegral}
\end{align}
and
\begin{equation}
\alpha_t^\to(n)= \frac{1}{2\pi} \int_{-\pi}^{\pi} \sum_j \xi_j(k) N_j(k) be^{ik} \lambda_j^t e^{-ikn}dk, \label{eq:rightintegral}
\end{equation}
Note that a discrete walk is being approximated by an integral. The Euler-Maclaurin summation formula\footnote{$\sum_{n=a}^b f(n) = \int_a^b f(x)dx + \frac{f(a)+f(b)}{2}+\sum_{k=1}^\infty \frac{B_{2k}}{(2k)!}(f^{(2k-1)}(b)-f^{(2k-1)}(a))$, where each $B_i$ is a Bernoulli number \cite{Apostol1999}.} gives the error term for these approximations.

Equations (\ref{eq:leftintegral}) and (\ref{eq:rightintegral}) can be solved by the steepest descent method from complex analysis, obtaining this way closed-form solutions. This is done in the next section.

\section{Asymptotic Approximation}\label{sec:approximation}
In this section it is shown how to find close-form solutions to the integrals (\ref{eq:leftintegral}) and (\ref{eq:rightintegral}). First, in Section \ref{sec:steepest-descent} the technique used in this research known as the steepest descent method is briefly explained. Then, in Section \ref{sec:asymptotic-approximation} the same technique is applied to the integral-forms of the walk (equations (\ref{eq:leftintegral}) and (\ref{eq:rightintegral})).

\subsection{Steepest Descent Method}\label{sec:steepest-descent}
Here one of the most powerful methods for asymptotic approximation of integrals is briefly explained. The method is known as Steepest Descent Approximation or Saddle Point Method. For a deeper understanding on this technique refer to \cite{wong01}.

The method of steepest descent is an asymptotic approximation method for certain types of exponential integrals of the form
\begin{equation}\label{eq:integral}
I_t=\int_C g(z)e^{tf(z)}dz
\end{equation}
where $C$ is a contour in the complex $z$-plane and $g(z)$ and $f(z)$ are complex-valued analytic functions. The parameter $t$ is taken to be real and positive, and we are interested in the asymptotic behavior of (\ref{eq:integral}) as $t\to\infty$ with $t>0$. Laplace's and stationary phase methods are just instances of this general procedure. The integral is dominated by the highest stationary points of $f$, i.e. if $f(z)=u(x,y)+iv(x,y)$ with $z=x+iy$ we expect the integral to be dominated by points where $u$ is maximum and $v$ is constant. The only possible extrema for $f$ are the {\it saddle points}, where $f'(z)=0$. Since $f$ is analytic, $u$ and $v$ satisfy the Cauchy-Riemann equation
\[
\frac{\partial^2u}{\partial x^2}+\frac{\partial^2u}{\partial y^2}=0,
\]
and from the maximum principle \cite{wong01} we have that if $\frac{\partial^2u}{\partial x^2}>0$ then $\frac{\partial^2u}{\partial y^2}<0$ or vice versa. If $z_0$ is the saddle point, then we can deform the contour to $C'$ (by Cauchy's theorem) so that it passes through $z_0$. From the Taylor expansion of $f(z)$ about $z_0$ we have
\[
f(z)\sim f(z_0)+\frac{1}{2}f''(z_0)(z-z_0)^2,
\]
where $\sim$ means ``is close up to additive error to". Then $g(z)\sim g(z_0)$, because for large $t$ the main contribution to the integral comes from $f$. Then $I_t$ becomes
\[
I_t\sim g(z_0)e^{tf(z_0)}\int_{C'}e^{\frac{1}{2}tf''(z_0)(z-z_0)^2}dz.
\]
Setting
\[
z-z_0=re^{i\phi} \quad\text{and}\quad f''(z_0)=\left|f''(z_0)\right|e^{i\theta}
\]
it can be seen that
\[
I_t\sim g(z_0)e^{tf(z_0)}\int_{C'}\exp({\frac{1}{2}t\left|f''(z_0)\right|e^{i\theta+2i\phi}r^2})e^{i\phi}dr.
\]

Note that $\phi$ is the angle of inclination of the oriented tangent to $C$ at point $z_0$, i.e.  $\phi=\arg(z_0)$ on $C$ \cite{wong01}. Choosing $\theta+2\phi=\pi$, i.e., $\phi=(\pi-\theta)/2$ then
\[
I_t\sim g(z_0)e^{tf(z_0)}e^{i\phi}\int_{C'}e^{-\frac{1}{2}t\left|f''(z_0)\right|r^2}dr
\]
and solving this as a Gaussian integral\footnote{The Gaussian integral or probability integral is given by $\int_{-\infty}^{\infty}e^{-x^2}dx=\sqrt{\pi}$.} yields
\begin{equation}\label{eq:expansion}
I_t= g(z_0)e^{tf(z_0)}e^{i\phi}\left(\frac{2\pi}{t\left|f''(z_0)\right|}\right)^{1/2}+O(t^{-1}).
\end{equation}

The deformation of the contour chosen to make the integration Gaussian corresponds to the steepest descent path from the saddle point, hence the name of the method \cite{wong01}. Taking this path is not essential, other methods like stationary point and Perron's method take another path with similar results \cite{wong01}.

\subsection{Asymptotic Approximation of the Walk on the Line}\label{sec:asymptotic-approximation}
\subsubsection{Left Amplitude}
First the integral-form corresponding to equation (\ref{eq:leftintegral}) is solved. First, put the integral in the form of equation (\ref{eq:integral}) by setting $n=\gamma t$ ($\gamma=n/t$) and writing
\begin{align}
\alpha_t^\gets(\gamma t)=\frac{1}{2\pi}\int_{-\pi}^{\pi} \sum_j g_j(k)e^{t f_j(k)}
\end{align}
where
\begin{align}
f_j(k)		&=\log{\lambda_j(k)}-ik\gamma,\label{eq:fj}\\
g_j(k)	&=N_j(k)\xi_j(k)\left(-ia\sin k\pm \sqrt{b^2-a^2\sin^2 k}\right).\label{eq:gj}
\end{align}

The saddle points $\theta_j$ of $f_j(k)$ are defined by the equation
\[
f_j'(\theta_j)=-i\gamma\mp\frac{a\sin \theta_j}{\sqrt{b^2-a^2\sin^2\theta_j}}=0.
\]
\redout{This equation has a solution at}
\begin{equation}\label{eq:theta}
\theta_j=\pm\arcsin \left(\frac{b\gamma}{a\sqrt{\gamma^2-1}}\right).
\end{equation}
\redout{Also note} that $|\lambda_j(\theta_j)|=1$. Moreover
\begin{align}\label{eq:f-theta}
f_j(\theta_j)=-i\gamma \theta_j+\log\left(\frac{\pm b+\sqrt{a^2(1-\gamma^2)+b^2\gamma^2}}{2\sqrt{1-\gamma^2}}\right)
\end{align}
and
\begin{align}\label{eq:f2-theta}
f_j''(\theta_j)=\frac{\pm(\gamma^2-1)\sqrt{b^2\gamma^2+a^2(1-\gamma^2)}}{b}.
\end{align}

\redout{Another solution to the equation $f'(\theta_j)=0$ is at $-\pi-\theta_j$ in the interval $[-\pi,\pi]$. However, since $f''(\theta_j)$ and $f''(-\pi-\theta_j)$ have similar behavior, the computations do not change.}

The contour is the real line in $[-\pi,\pi]$ and has no imaginary part, therefore $\phi=\arg \theta_j =0$ in equation (\ref{eq:expansion}).

Now using (\ref{eq:expansion}), the asymptotic expansion can be obtained
\begin{align*}
\alpha^\gets_t(\gamma t)		&= \frac{1}{2\pi} \sum_j g_j(\theta_j)e^{tf_j(\theta_j)} \left( \frac{2\pi}{t|f''_j(\theta_j)|} \right)^{1/2} + O(t^{-1}) \nonumber\\
			&=\frac{1}{2\pi} \sum_j N_j(\theta_j) \xi_j(\theta_j) \left[ \frac{\pm b(1-\gamma)}{\sqrt{1-\gamma^2}} \right] \nonumber\\
			&\quad\times\left( \frac{\pm b+\sqrt{a^2(1-\gamma^2)+b^2\gamma^2}}{2\sqrt{1-\gamma^2}} \right)^t e^{-i\gamma\theta_j t} \nonumber\\
			&\quad \times\left( \frac{2\pi |b|}{t |\gamma^2-1| \sqrt{b^2\gamma^2+a^2(1-\gamma^2)}} \right)^{1/2} + O(t^{-1}).
\end{align*}

\subsubsection{Right Amplitude}
Next is the solution of equation (\ref{eq:rightintegral}). Following the same steps as above, write the integral as
\begin{align}
\alpha_t^\to(\gamma t)=\frac{1}{2\pi}\int_{-\pi}^{\pi} \sum_j h_j(k)e^{t f_j(k)},
\end{align}
where $f_j$ is defined in the same way as in (\ref{eq:fj}), and
\begin{equation}
h_j(k)=N_j(k)\xi_j(k) b e^{ik}.\label{eq:hj}
\end{equation}

Reusing the previous calculations for $f_j$ (equations (\ref{eq:theta}), (\ref{eq:f-theta}) and (\ref{eq:f2-theta})), the asymptotic expansion is
\begin{align*}
\alpha_t^\to(\gamma t)	&= \frac{1}{2\pi} \sum_j h_j(\theta_j)e^{t f_j(\theta_j)} \left( \frac{2\pi}{t|f_j''(\theta_j)|} \right)^{1/2} + O(t^{-1})\nonumber\\
				&=\frac{1}{2\pi} \sum_j N_j(\theta_j)\xi_j(\theta_j) b e^{i\theta_j} \nonumber\\
				&\quad\times\left( \frac{\pm b+\sqrt{a^2(1-\gamma^2)+b^2\gamma^2}}{2\sqrt{1-\gamma^2}} \right)^t e^{-i\gamma\theta_j t}\nonumber\\
				&\quad \times\left( \frac{2\pi|b|}{t |\gamma^2-1|\sqrt{b^2\gamma^2+a^2(1-\gamma^2)}|} \right)^{1/2}+O(t^{-1})
\end{align*}

\section{Closed-form Formulas and Convergence}\label{sec:formulas-convergence}
\subsection{Formulas}
Approximate closed-forms for the amplitudes of the state of the walk on the line were given. Now the main contribution of this paper can be stated formally.
\begin{theorem}\label{the:amplitudes}
Let $\gamma=n/t$ and $a\equiv e^{i\pi\tau_1}+e^{i\pi\tau_2},b\equiv e^{i\pi\tau_1}-e^{i\pi\tau_2}$. If the state of the walk is
\[ \ket{\Psi_t}=\sum_n \ket{\psi_t(n)} \text{ with } \ket{\psi_t(n)}=
\begin{bmatrix}
\alpha^\gets_t(n)\\
\alpha_t^\to(n)
\end{bmatrix}
\]
then,
\begin{align*}
\alpha^\gets_t(\gamma t)	&\sim\frac{1}{2\pi} \sum_j N_j \xi_jA_j\left[ \frac{\pm b(1-\gamma)}{\sqrt{1-\gamma^2}} \right],\\
\alpha_t^\to(\gamma t)	&\sim\frac{1}{2\pi} \sum_j N_j\xi_j A_j b e^{i\theta_j},
\end{align*}
where the terms $A_j, N_j, \xi_j$ and $\theta_j$ are given by
\begin{align*}
A_j					&=\left( \frac{\pm b+\sqrt{a^2(1-\gamma^2)+b^2\gamma^2}}{2\sqrt{1-\gamma^2}} \right)^t\\
					&\quad\times\left( \frac{2\pi |b|}{t |\gamma^2-1| \sqrt{b^2\gamma^2+a^2(1-\gamma^2)}} \right)^{1/2}e^{-i\gamma\theta_j t},\\
N_j					&=\left( \left|-ia\sin \theta_j \pm \sqrt{b^2-a^2\sin^2 \theta_j}\right|^2+\left| b\right|^2 \right),\\
\xi_j					&=\alpha_0^\gets(0)\left(-ia\sin \theta_j \pm \sqrt{b^2-a^2\sin^2\theta_j}\right)^*\\
					&\quad+\alpha_0^\to(0)b^*e^{-i\theta_j},\\
\sin\theta_j			&=\pm \left(\frac{b\gamma}{a\sqrt{\gamma^2-1}}\right),
\end{align*}
with $\alpha_0^\gets(0)$ and $\alpha_0^\to(0)$ as the initial amplitudes of the walk for $n=0$, and $\alpha_0^\gets(n)=\alpha_0^\to(n)=0$ for $n \neq 0$.
\end{theorem}

\begin{figure}[t]
	\centering
	\includegraphics[scale=1.2]{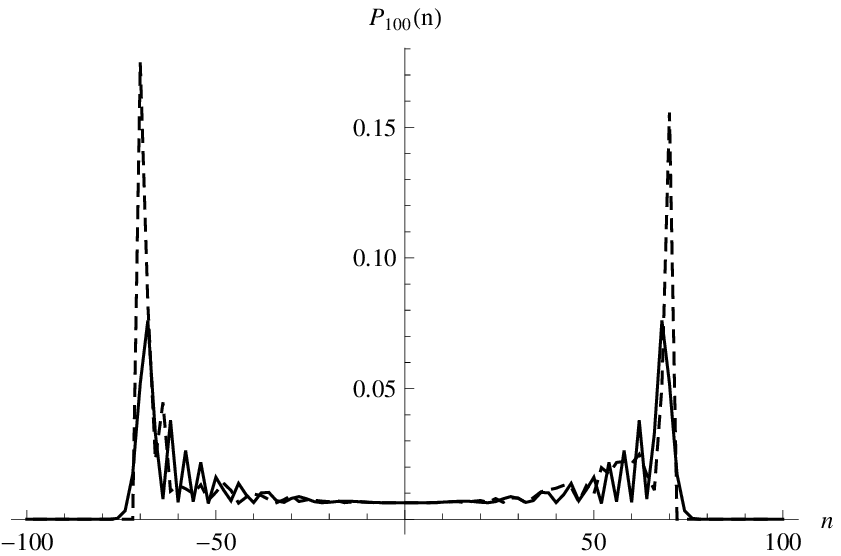}
	\caption{Comparison between the probability distributions of numerical simulation (dark) and Theorem \ref{the:amplitudes} (dashed) with $\tau_1=1/2$ and $\tau_2=0$, $t=100$, and initial state in equal superposition of directions.}
	\label{fig:comparison1}
\end{figure}
\begin{figure}[t]
	\centering
	\includegraphics[scale=1.2]{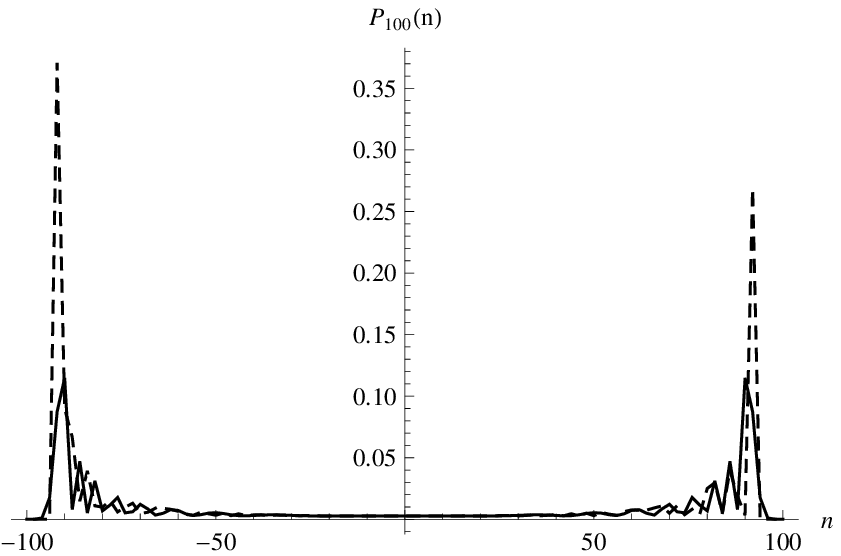}
	\caption{Comparison between the probability distributions of numerical simulation (dark) and Theorem \ref{the:amplitudes} (dashed) with $\tau_1=3/4$ and $\tau_2=1/2$, $t=100$, and initial state in equal superposition of directions.}
	\label{fig:comparison2}
\end{figure}

In a seminal work, Konno \cite{Konno2002,Konno2005} gave explicit expressions for the amplitudes of a $U(2)$ coin using a discrete path integral method. However, these expressions were not in closed-form, as it is claimed in this work.


In order to assess the quality of the approximation, figures \ref{fig:comparison1} and \ref{fig:comparison2} show a comparison between the probability distributions given by Theorem \ref{the:amplitudes}, and a numerical simulation of walks that start with an equal superposition of directions for different values of the parameters. It can be seen that the approximation gives some errors, but the asymptotic agrees with the simulation. The figures show that Theorem \ref{the:amplitudes} is close to the real values of the probability distribution, in particular in the middle part of the plots.

The errors in the approximation made by Theorem \ref{the:amplitudes} can be computed from two parts, the Euler-Maclaurin formula and the steepest descent method \cite{wong01}. Denote these errors by $\epsilon$ and $\varepsilon$ respectively. Let $B_i=\sum_{r=0}^i\binom{i}{r}B_{i-r}$ be a Bernoulli number \cite{Apostol1999}, and let $d\in\{\gets,\to\}$. Then, the error for $\alpha_t^d(\gamma t)$ is $\sum_j \epsilon_{j,d}+\varepsilon_{j,d}$, where
\begin{equation}
\epsilon_{j,d}=\sum_{m=1}^\infty \frac{B_{2m}}{(2m)!} \left(\frac{\partial^{2m-1}}{\partial k^{2m-1}}\widetilde{\alpha}_t^{d}(\pi)-\frac{\partial^{2m-1}}{\partial k^{2m-1}}\widetilde{\alpha}_t^{d}(-\pi) \right)
\end{equation}
and
\begin{align}
\varepsilon_{j,d}		&=\frac{1}{2\pi} \sum_j e^{t f_j(\theta_j)} \left( \frac{2\pi}{t|f_j''(\theta_j)|} \right)^{1/2}\nonumber\\					&\quad\times\left(\sum_{m=1}^\infty  \frac{(-1)^m}{m!} \left( \frac{1}{2t|f_j''(\theta_j)|} \right)^m \frac{\partial^{2m}}{\partial k^{2m}} \rho_j(\theta_j) \right),
\end{align}
where $\rho_j$ is either equation (\ref{eq:gj}) if $d=\gets$, or (\ref{eq:hj}) if $d=\to$. It can be seen that if we take $m$ terms from each summation, $\epsilon_{j,d}=O(2^{-m})$ and $\varepsilon_{j,d}=O(t^{-m})$.

\subsection{Convergence and Properties}\label{sec:convergence}
For quantum walks on the line and $n$-dimensional grids there exists weak convergence theorems  \cite{grimmett04}. In this section, we state the weak convergence of quantum walks on the line with phase parameters using these previous results. Then we show some applications of the convergence to compute the support of the probability density function.
\begin{theorem}\label{the:convergence}
Let $\Omega=[-\pi,\pi] \times\{1,2\}$ be a probability space with probability measure  $\mu=|\langle \widetilde{\psi}_0(k)|\lambda_j(k)\rangle |^2dk/2\pi$ for $k\in[-\pi,\pi]$ and $j=1,2$. Define a map $h:\Omega\to \mathbb{R}$ such that for $(k,j)\in \Omega$
\[ h(k,j)\equiv h_j(k)=(-1)^j \frac{\sin k}{\sqrt{\sin^2 k+\tan^2 \frac{\pi}{2}(\tau_1-\tau_2)}}.\]
Let $X_t$ be a position of the quantum walk at time $t$ with distribution given by (\ref{eq:prob}), and $Z$ be a random variable of $\Omega$ with distribution $\mu$. Then we have as $t\to\infty$
\[ \frac{X_t}{t} \Rightarrow h(Z),\]
where $\Rightarrow$ denotes weak convergence\footnote{A sequence of random variables $\{X_i:i\geq 1\}$ converges weakly to a random variable $Z$ if $\lim_{n\to\infty} X_n=Z$, given that $\lim_{n\to\infty} E[h(X_n)]=E[h(Z)]$ for all bound continuous functions $h:\mathbb{R}\to\mathbb{R}$.}.
\end{theorem}
\begin{proof}
Consider the theorem that states the weak convergence of quantum walks on the line \cite[theorem 1]{grimmett04}. Let $\lambda_j(k)$ be as in Lemma \ref{the:spectrum}. Then
\[
\lambda_j'(k)=\frac{-a\sin k}{2}-\frac{a^2\cos k \sin k}{2\sqrt{b^2-a^2\sin^2 k}}.
\]
Dividing this by $\lambda_j(k)$ we obtain
\[
\frac{-i\lambda_j'(k)}{\lambda_j(k)}=(-1)^{j+1} \frac{a i\sin k}{\sqrt{b^2-a^2 \sin^2(k)}}.
\]
Then, after some algebra and observing that $\frac{b}{a}=e^{i\pi/2}\tan \frac{\pi}{2}(\tau_1-\tau_2)$, the theorem follows.
\end{proof}
As an application of Theorem \ref{the:convergence}, we can calculate the position of the two peaks of the walk for large time.
\begin{corollary}\label{cor:interval}
The limit distribution of $X_t/t$ is concentrated on the interval $\left[ -  \frac{|a|}{2}  ,    \frac{|a|}{2}   \right]$.
\end{corollary}
\begin{proof}
Theorem \ref{the:convergence} have its maximum and minimum values for $k=\pm\pi/2$ and the corollary follows.
\end{proof}

The maximum probability of $P_t(n)$ is found at the top of these two peaks, i.e., where \redout{$n=\pm   |a|/2$} \cite{grimmett04}. Considering $|n/t|$ as the speed of the peaks, it can be seen that by setting $\tau_1=\tau_2$ it gets its maximum value, i.e. the fastest spreading of the walk. This corresponds exactly to an identity operator, and the walk does not mix at all inside the range of corollary \ref{cor:interval}. In order to get high speed and maximum randomness (i.e. the best mixing for positions inside the range) for $P_t(n)$, we can set any value such that  $|\tau_1-\tau_2|=1/2$. This implies that the support of $h$ is in $[-1/\sqrt{2},1/\sqrt{2}]$. In this case, the operator simulates exactly the probability distribution of a Hadamard operator \cite{grimmett04}.

As another application of Theorem \ref{the:convergence}, we can compute the density function of the random variable $Y=X_t/t$ in the asymptotic limit when $t\to \infty$. Following the steps of \cite{grimmett04} for the Hadamard coin, we differentiate the quantity
\begin{equation}
P(Y \leq y)=\sum_j \int_{k\in[-\pi,\pi]:h_j(k)\leq y} \left| \braket{\widetilde{\psi}_0(k)}{\lambda_j(k)} \right|^2 \frac{dk}{2\pi},
\end{equation}
\redout{which yields the density function
\begin{equation}
f(y)=\frac{|b|/2}{\pi(y^2-1)\sqrt{(|a|/2)^2-y^2}}
\end{equation}
for $y\in(-|a|/2,|a|/2)$, under the assumption of $Im(\alpha_0^\gets\cdot \alpha_0^{\to *})\sin(\tau_1-\tau_2)\pi=0$ and $|\alpha_0^\gets|=|\alpha_0^\to|=1/\sqrt{2}$, which agrees with \cite{Konno2002,Konno2005}.}

\section{Conclusions}
This paper presented a study of discrete-time quantum walks on the line. A symmetric $SU(2)$ coin operation was proposed and analyzed as a step towards an understanding of quantum walks. Using Fourier analysis and asymptotic approximation methods, we computed a closed-form formula for the amplitudes of the state of the walk. With this formula, we have a direct way to compute the amplitudes at any time step without recurring to time-consuming simulations or numerical integration. This also give us a complete characterization of the induced probability distribution of general quantum walks on the line.

One important question that remains unanswered is the relation between Theorems \ref{the:amplitudes} and \ref{the:convergence}. Theorem \ref{the:amplitudes} is based on the computation of saddle points of the high oscillatory kernel of Fourier coefficients. On the other hand, Theorem \ref{the:convergence} is based on the method of moments (see \cite{grimmett04} for details). A relation between these two density functions could set a common ground for the analysis of coined quantum walks.

\section*{Aknowledgments}
The authors would like to thank Hiroyuki Seki for reviewing an initial draft. The authors also thank the anonymous referees for the suggestions that helped improve this paper.

\bibliographystyle{plain}
\bibliography{library}

\begin{thebibliography}{10}

\bibitem{ambainis04}
Andris Ambainis.
\newblock {Quantum Walks and their Algorithmic Applications}.
\newblock {\em International Journal of Quantum Information}, 1(4):507--518,
  2004.

\bibitem{ambainis07}
Andris Ambainis.
\newblock {Quantum Walk Algorithm for Element Distinctness}.
\newblock {\em SIAM Journal on Computing}, 37(1):210--239, 2007.

\bibitem{ambainis01}
Andris Ambainis, Eric Bach, Ashwin Nayak, Ashvin Vishwanath, and John Watrous.
\newblock {One-Dimensional Quantum Walks}.
\newblock In {\em Proceedings of the 33rd Annual ACM Symposium on Theory of
  Computing}, pages 37--49, 2001.

\bibitem{ambainis05}
Andris Ambainis, Julia Kempe, and Alexander Rivosh.
\newblock {Coins Make Quantum Walks Faster}.
\newblock In {\em Proceedings of the 16th ACM-SIAM Symposium on Discrete
  Algorithms}, pages 1099--1102, 2005.

\bibitem{Apostol1999}
Tom Apostol.
\newblock {An Elementary View of Euler's Summation Formula}.
\newblock {\em The American Mathematical Monthly}, 106(5):409--418, 1999.

\bibitem{chandrashekar08}
C.~M. Chandrashekar, R.~Srikanth, and Raymond Laflamme.
\newblock {Optimizing the Discrete Time Quantum Walk using a SU(2) Coin}.
\newblock {\em Physical Review A}, 77(032326), 2008.

\bibitem{Graham1994}
Ronald Graham, Donald Knuth, and Oren Patashnik.
\newblock {\em {Concrete Mathematics: A Foundation for Computer Science}}.
\newblock Addison-Wesley Professional, 2nd edition, 1994.

\bibitem{grimmett04}
Geoffrey Grimmett, Svante Janson, and Petra Scudo.
\newblock {Weak Limits for Quantum Random Walks}.
\newblock {\em Physical Review E}, 69(026119), 2004.

\bibitem{Hogg2000}
Tad Hogg.
\newblock {Quantum search heuristics}.
\newblock {\em Physical Review A}, 61(5), April 2000.

\bibitem{kempe03}
Julia Kempe.
\newblock {Quantum Random Walks: An Introductory Overview}.
\newblock {\em Contemporary Physics}, 44(4):307--327, 2003.

\bibitem{kempe05}
Julia Kempe.
\newblock {Discrete Quantum Walk Hit Exponentially Faster}.
\newblock {\em Probability Theory and Related Fields}, 133(2), 2005.

\bibitem{Konno2002}
Norio Konno.
\newblock {Quantum Random Walks in One Dimension}.
\newblock {\em Quantum Information Processing}, 1(5):345--354, 2003.

\bibitem{Konno2005}
Norio Konno.
\newblock {A new type of limit theorems for the one-dimensional quantum random
  walk}.
\newblock {\em Journal of the Mathematical Society of Japan}, 57(4):1179--1195,
  2005.

\bibitem{Konno2008}
Norio Konno.
\newblock {Quantum Walks}.
\newblock {\em Lecture Notes in Mathematics}, 1954:309--452, 2008.

\bibitem{moore02}
Cristopher Moore and Alexander Russell.
\newblock {Quantum Walks on the Hypercube}.
\newblock In {\em Proceedings of the 6th International Workshop on
  Randomization and Approximation Techniques}, pages 164--178, London, UK,
  2002. Springer-Verlag.

\bibitem{Nayak2000}
Ashwin Nayak and Ashvin Vishwanath.
\newblock {Quantum Walk on the Line}.
\newblock In {\em arXiv:quant-ph/0010117}, 2000.

\bibitem{nielsen00}
Michael Nielsen and Isaac Chuang.
\newblock {\em {Quantum Computation and Quantum Information}}.
\newblock Cambridge University Press, 2000.

\bibitem{potocek09}
V\'{a}clav Poto\v{c}ek, Aur\'{e}l G\'{a}bris, Tam\'{a}s Kiss, and Igor Jex.
\newblock {Optimized Quantum Random-Walk Search Algorithms on the Hypercube}.
\newblock {\em Physical Review A}, 79(012325), 2009.

\bibitem{santha08}
Miklos Santha.
\newblock {Quantum Walk Based Search Algorithms}.
\newblock In {\em Proceedings of the 5th international conference on Theory and
  applications of models of computation}, pages 31--46, 2008.

\bibitem{shenvi03}
Neil Shenvi, Julia Kempe, and Birgitta Whaley.
\newblock {Quantum Random-Walk Search Algorithm}.
\newblock {\em Physical Review A}, 67(052307), 2003.

\bibitem{Stefanak2008}
Martin Stefanak, Tamas Kiss, and Igor Jex.
\newblock {Recurrence properties of unbiased coined quantum walks on infinite
  d-dimensional lattices}.
\newblock {\em Physical Review A}, 78(032306), 2008.

\bibitem{stroock05}
Daniel Stroock.
\newblock {\em {An Introduction to Markov Processes}}.
\newblock Springer, 2005.

\bibitem{wong01}
Roderick Wong.
\newblock {\em {Asymptotic Approximation of Integrals}}.
\newblock SIAM: Society for Industrial and Applied Mathematics, 2001.

\end{thebibliography}

\end{document}